\documentclass[11pt]{article}
\usepackage{amssymb}

\usepackage{graphicx}
\usepackage{amsmath}
\usepackage{makeidx}
\usepackage{indentfirst}

\newcounter{resultnum}[section]

\newcounter{conclusionnum}[section]

\newcounter{conditionnum}[section]

\newcounter{conjecturenum}[section]

\newcounter{examplenum}[section]

\newcounter{exercisenum}[section]
\newtheorem{lemma}{Lemma}[section]

\newcounter{lemmanum}[section]

\newcounter{notationnum}[section]
\newtheorem{theorem}{Theorem}[section]

\newcounter{theoremnum}[section]
\newtheorem{definition}{Definition}[section]

\newcounter{definitionnum}[section]

\newcounter{corollarynum}[section]

\newcounter{remarknum}[section]
\newtheorem{proposition}{Proposition}[section]

\newcounter{propositionnum}[section]

\newcounter{acknowledgementnum}[section]

\newcounter{algorithmnum}[section]

\newcounter{axiomnum}[section]

\newcounter{casenum}[section]

\newcounter{claimnum}[section]

\newcounter{summarynum}[section]

\newcounter{problemnum}[section]
\newenvironment{proof}[1][]{\textbf{Proof.} }{}

\begin{document}

\title{Fedosov Quantization of Fractional Lagrange Spaces }
\date{July 4, 2010}
\author{\textbf{Dumitru Baleanu}\thanks{%
On leave of absence from Institute of Space Sciences, P. O. Box, MG-23, R
76900, Magurele--Bucharest, Romania, \newline
E--mails: dumitru@cancaya.edu.tr, baleanu@venus.nipne.ro} \\
\textsl{\small Department of Mathematics and Computer Sciences,} \\
\textsl{\small \c Cankaya University, 06530, Ankara, Turkey } \\
\and 
\textbf{Sergiu I. Vacaru} \thanks{%
sergiu.vacaru@uaic.ro, Sergiu.Vacaru@gmail.com \newline
http://www.scribd.com/people/view/1455460-sergiu} \and \textsl{\small %
Science Department, University "Al. I. Cuza" Ia\c si, } \\
\textsl{\small 54, Lascar Catargi street, Ia\c si, Romania, 700107 } }
\maketitle

\begin{abstract}
The main goal of this work is to perform a nonolonomic deformation (Fedosov
type) quantization of fractional Lagrange geometries. The constructions are
provided for a (fractional) almost K\"{a}hler model encoding equivalently
all data for fractional Euler--Lagrange equations with Caputo fractional
derivative. For homogeneous generating Finsler functions, the geometric
models contain quantum versions of fractional Finsler spaces. The scheme can
be generalized for fractional Hamilton systems and various models of
fractional classical and quantum gravity. We conclude that the approach with
Caputo fractional derivative allows us to geometrize both classical and
quantum (Fedosov type) fractional regular Lagrange interactions. \vskip0.2cm
\textbf{Keywords:}\ fractional Lagrange and Finsler geometry, almost K\"{a}%
hler spaces, nonlinear connections, deformation quantization, fractional
Fedosov spaces.

\vskip3pt


\end{abstract}


\section{Introduction}

This paper is a ''quantum'' partner work (in the meaning of deformation
quantization) of the article \cite{bv1} on almost K\"{a}hler models of
fractional Lagrange--Finsler geometries. It belongs to a series of our works
on fractional (i.e. non--integer dimension) nonholonomic spaces, theirs
Ricci flows and certain fractional type gravity and geometric mechanics
models \cite{vrfrf,vrfrg}\footnote{%
readers are recommended to consult the main results and conventions in
advance}.

There are some preliminary attempts (see, for instance, \cite%
{lask1,lask2,nab04}) to quantize fractional mechanical models and field
interactions as generalized quantum mechanics and related fractional quantum
field theories. Such constructions are for some particular cases of
fractional calculus and physical models. It is not clear if and how a
general formalism encoding quantum fractional theories can be elaborated.

Following the geometry of nonholonomic distributions modelling geometries of
non--integer (i.e. fractional) dimensions, a self--consistent quantization
formalism can be elaborated following the Fedosov deformation quantization %
\cite{fed1,fed2}. The original constructions were provided for classical and
quantum K\"{a}hler geometries. Latter, the approach was generalized for
almost K\"{a}hler geometries \cite{karabeg1} which allowed to include into
the quantization deformation scheme various types of Lagrange--Finsler,
Hamilton--Cartan and Einstein spaces and generalizations, see papes \cite%
{vfedq1,vfedq2,vfedq3,vfed4} and references therein.

A geometrization of fractional calculus and various types of fractional
mechanical and field theories is possible for the so--called Caputo
fractional derivative, see details in \cite{vrfrf,vrfrg}. For fractional
regular Lagrange mechanical \ models, such effective geometries can be
derived as almost K\"{a}hler configurations with more ''rich'' geometric
nonholonomic and/or fractional structures \cite{bv1}. This suggests a
''realistic'' possibility to quantize in general form, following methods of
deformation quantiation, various types of \ fractional geometries and
physical theories which via nonholonomic deformations can be re--defined as
some types of almost K\"{a}hler spaces. In this article, we show how to
perform such a program for fractional Lagrange spaces.

This work is organized in the form: In section \ref{s2}, we remember the
most important properties and formulas on Caputo fractional derivatives and
related nonholonomic (co) frame formalism for fractional tangent bundles.
Section \ref{s3} is devoted to definition of Fedosov operators for
fractional Lagrange spaces. The main results on deformation quantization of
fractional Lagrange mechanics are provided in section \ref{s4}. Finally, we
conclude and discuss the results in section \ref{s5}.

\section{Preliminaries:\ Fractional Calculus and almost K\"{a}hler Geometry}

\label{s2}We outline some necessary formulas and results on fractional
calculus and nonholonomic geometry, see details and notation conventions in %
\cite{bv1,vrfrf,vrfrg}.\footnote{%
We use ''up'' and ''low'' left labels which are convenient to be introduced
in order to not create confusions with a number of ''horizontal'' and
''vertical'' right indices and labels which must be distinguished if the
manifolds are provided with N--connection structure. In our papers, we work
with mixed sets of "fractional" and "integer" dimensions (and holonomic and
nonholonomic variables etc). This makes the systems of labels and notations
for geometric objects to be quite sophisticate even in coordinate free form
formalisms. Unfortunately, further simplifications seem to be not possible.}

\subsection{Caputo fractional derivatives}

There is a class of fractional derivatives which resulting in zero for
actions on constants. This property is crucial for constructing geometric
models of theories with fractional calculus.

We define and denote the fractional left, respectively, right Caputo
derivatives in the forms
\begin{eqnarray}
&&\ _{\ _{1}x}\overset{\alpha }{\underline{\partial }}_{x}f(x):=\frac{1}{%
\Gamma (s-\alpha )}\int\limits_{\ \ _{1}x}^{x}(x-\ x^{\prime })^{s-\alpha
-1}\left( \frac{\partial }{\partial x^{\prime }}\right) ^{s}f(x^{\prime
})dx^{\prime };  \label{lfcd} \\
&&\ _{\ x}\overset{\alpha }{\underline{\partial }}_{\ _{2}x}f(x):=\frac{1}{%
\Gamma (s-\alpha )}\int\limits_{x}^{\ _{2}x}(x^{\prime }-x)^{s-\alpha
-1}\left( -\frac{\partial }{\partial x^{\prime }}\right) ^{s}f(x^{\prime
})dx^{\prime }\ .  \notag
\end{eqnarray}%
The fractional absolute differential $\overset{\alpha }{d},$ corresponding
to above fractional derivatives, is written $\overset{\alpha }{d}%
:=(dx^{j})^{\alpha }\ \ _{\ 0}\overset{\alpha }{\underline{\partial }}_{j},$
where $\ \overset{\alpha }{d}x^{j}=(dx^{j})^{\alpha }\frac{(x^{j})^{1-\alpha
}}{\Gamma (2-\alpha )},$ where we consider $\ _{1}x^{i}=0.$

For a fractional tangent bundle $\overset{\alpha }{\underline{T}}M$ \ for $%
\alpha \in (0,1),$ associated to a manifold $M$ of necessary smooth class
and integer $\dim M=n,$ we write both the integer and fractional local
coordinates in the form $u^{\beta }=(x^{j},y^{a}).$ The symbol $T$ is
underlined in order to emphasize that we shall associate the approach to a
fractional Caputo derivative. \ A fractional frame basis $\overset{\alpha }{%
\underline{e}}_{\beta }=e_{\ \beta }^{\beta ^{\prime }}(u^{\beta })\overset{%
\alpha }{\underline{\partial }}_{\beta ^{\prime }}$ on $\overset{\alpha }{%
\underline{T}}M$ \ is connected via a vielbein transform $e_{\ \beta
}^{\beta ^{\prime }}(u^{\beta })$ with a fractional local coordinate basis
\begin{equation}
\overset{\alpha }{\underline{\partial }}_{\beta ^{\prime }}=\left( \overset{%
\alpha }{\underline{\partial }}_{j^{\prime }}=_{\ _{1}x^{j^{\prime }}}%
\overset{\alpha }{\underline{\partial }}_{j^{\prime }},\overset{\alpha }{%
\underline{\partial }}_{b^{\prime }}=_{\ _{1}y^{b^{\prime }}}\overset{\alpha
}{\underline{\partial }}_{b^{\prime }}\right) ,  \label{frlcb}
\end{equation}%
for $j^{\prime }=1,2,...,n$ and $b^{\prime }=n+1,n+2,...,n+n.$ The
fractional co--bases are written $\overset{\alpha }{\underline{e}}^{\ \beta
}=e_{\beta ^{\prime }\ }^{\ \beta }(u^{\beta })\overset{\alpha }{d}u^{\beta
^{\prime }},$ where the fractional local coordinate co--basis is
\begin{equation}
\ _{\ }\overset{\alpha }{d}u^{\beta ^{\prime }}=\left( (dx^{i^{\prime
}})^{\alpha },(dy^{a^{\prime }})^{\alpha }\right) .  \label{frlccb}
\end{equation}

Explicit constructions in the geometry of fractional tangent bundle depend
on the type of chosen fractional derivative.

\subsection{A geometrization of fractional Lagrange mechanics}

A fractional Lagrange space $\overset{\alpha }{\underline{L^{n}}}=(\overset{%
\alpha }{\underline{M}},\overset{\alpha }{L})$ of fractional dimension $%
\alpha \in (0,1),$ for a regular real function $\overset{\alpha }{L}:$ $%
\overset{\alpha }{\underline{T}}M\rightarrow \mathbb{R},$ is associated to a
prime Lagrange space $L^{n}=(M,L),$ of integer dimension $n,$ which (in its
turn) is defined by a Lagrange fundamental function $L(x,y),$ i.e. a regular
real function $L:$ $TM\rightarrow \mathbb{R},$ for which the Hessian $%
_{L}g_{ij}=(1/2)\partial ^{2}L/\partial y^{i}\partial y^{j}$ is not
generated.

Any $\overset{\alpha }{L}(x,\ ^{\alpha }y)$ determines three fundamental
geometric objects on $\overset{\alpha }{\underline{L^{n}}}:$

\begin{enumerate}
\item A canonical N--connection $\ \ _{L}\overset{\alpha }{\mathbf{N}}%
\mathbf{=}\{\ _{L}^{\alpha }N_{i}^{a}\}$ structure \ (with local coefficients%
\footnote{%
computed with the aim to encode the fractional Euler--Lagrange equations
into a canonical semi--spray configuration \cite{bv1}, $\ $%
\begin{equation*}
\ _{L}^{\alpha }N_{j}^{a}=\ _{\ _{1}y^{j}}\overset{\alpha }{\underline{%
\partial }}_{j}\overset{\alpha }{G^{k}}(x,\ ^{\alpha }y)\mbox{\ for \ }%
\overset{\alpha }{G^{k}}=\frac{1}{4}\ \ _{L\ }\overset{\alpha }{g^{kj}}\left[
y^{j}\ _{\ _{1}y^{j}}\overset{\alpha }{\underline{\partial }}_{j}\ \left(
_{\ _{1}x^{i}}\overset{\alpha }{\underline{\partial }}_{i}\overset{\alpha }{L%
}\right) -\ _{\ _{1}x^{i}}\overset{\alpha }{\underline{\partial }}_{i}%
\overset{\alpha }{L}\right]
\end{equation*}%
} $\ _{L}^{\alpha }N_{i}^{a}$ parametrized for a decomposition $\ _{L}%
\overset{\alpha }{\mathbf{N}}\mathbf{=}\ \ _{L}^{\alpha
}N_{i}^{a}(u)(dx^{i})^{\alpha }\otimes \overset{\alpha }{\underline{\partial
}}_{a}$ with respect to local bases (\ref{frlcb}) and (\ref{frlccb})) with
an associated class of \ N--adapted fractional (co) frames linearly
depending on $\ ^{\alpha }N_{i}^{a},$
\begin{eqnarray}
\ _{L}^{\alpha }\mathbf{e}_{\beta } &=&\left[ \ _{L}^{\alpha }\mathbf{e}_{j}=%
\overset{\alpha }{\underline{\partial }}_{j}-\ _{L}^{\alpha }N_{j}^{a}%
\overset{\alpha }{\underline{\partial }}_{a},\ ^{\alpha }e_{b}=\overset{%
\alpha }{\underline{\partial }}_{b}\right] ,  \label{dder} \\
\ _{L}^{\alpha }\mathbf{e}^{\beta } &=&[\ ^{\alpha }e^{j}=(dx^{j})^{\alpha
},\ _{L}^{\alpha }\mathbf{e}^{b}=(dy^{b})^{\alpha }+\ _{L}^{\alpha
}N_{k}^{b}(dx^{k})^{\alpha }].  \label{ddif}
\end{eqnarray}

\item A canonical (Sasaki type) metric structure,
\begin{eqnarray}
\ \ _{L}\overset{\alpha }{\mathbf{g}} &=&\ _{L}^{\alpha }g_{kj}(x,y)\
^{\alpha }e^{k}\otimes \ ^{\alpha }e^{j}+\ _{L}^{\alpha }g_{cb}(x,y)\
_{L}^{\alpha }\mathbf{e}^{c}\otimes \ _{L}^{\alpha }\mathbf{e}^{b},
\label{sasm} \\
\ _{L\ }\overset{\alpha }{g}_{ij} &=&\frac{1}{4}\left( \overset{\alpha }{%
\underline{\partial }}_{i}\overset{\alpha }{\underline{\partial }}_{j}+%
\overset{\alpha }{\underline{\partial }}_{j}\overset{\alpha }{\underline{%
\partial }}_{i}\right) \overset{\alpha }{L}\neq 0,  \notag
\end{eqnarray}%
with $\ _{L}^{\alpha }g_{cb}$ computed respectively by the same formulas as $%
\ _{L}^{\alpha }g_{kj}.$

\item A canonical fractional metrical d--connec\-ti\-on $\ _{c}^{\alpha }%
\mathbf{D=}(h\ _{c}^{\alpha }D,v\ _{c}^{\alpha }D)=\{\ ^{\alpha }\widehat{%
\mathbf{\Gamma }}_{\ \alpha \beta }^{\gamma }=(\ ^{\alpha }\widehat{L}_{\
jk}^{i},\ ^{\alpha }\widehat{C}_{jc}^{i})\}\mathbf{,}$\textbf{\ }where%
\begin{eqnarray}
\ ^{\alpha }\widehat{L}_{jk}^{i} &=&\frac{1}{2}\ _{L}^{\alpha }g^{ir}\left(
\ _{L}^{\alpha }\mathbf{e}_{k}\ _{L}^{\alpha }g_{jr}+\ _{L}^{\alpha }\mathbf{%
e}_{j}\ _{L}^{\alpha }g_{kr}-\ _{L}^{\alpha }\mathbf{e}_{r}\ _{L}^{\alpha
}g_{jk}\right) ,  \label{cdc} \\
\ \ ^{\alpha }\widehat{C}_{bc}^{a} &=&\frac{1}{2}\ _{L}^{\alpha
}g^{ad}\left( \ ^{\alpha }e_{c}\ _{L}^{\alpha }g_{bd}+\ ^{\alpha }e_{c}\
_{L}^{\alpha }g_{cd}-\ ^{\alpha }e_{d}\ _{L}^{\alpha }g_{bc}\right)  \notag
\end{eqnarray}%
for $\ _{L}^{\alpha }g^{ad}$ being inverse to $\ _{L}^{\alpha }g_{kj}.$
\end{enumerate}

We conclude that the regular fractional mechanics defined by a fractional
Lagrangian $\overset{\alpha }{L}$ can be equivalently encoded \ into
canonical geometric data $\left( \ \ _{L}\overset{\alpha }{\mathbf{N}},\ \
_{L}\overset{\alpha }{\mathbf{g}},\ _{c}^{\alpha }\mathbf{D}\right) .$ This
allows us to apply a number of powerful geometric methods in \ fractional
calculus and applications.

\subsection{An almost K\"{a}hler--Lagrange model of fractional mechanics}

A fractional nonholonomic almost complex structure can be defined as a
linear operator $\overset{\alpha }{\mathbf{J}}$ acting on the vectors on $%
\overset{\alpha }{\underline{T}}M$ following formulas
\begin{equation*}
\overset{\alpha }{\mathbf{J}}(\ _{L}^{\alpha }\mathbf{e}_{i})=-\ ^{\alpha
}e_{i}\mbox{\ and \ }\overset{\alpha }{\mathbf{J}}(\ ^{\alpha }e_{i})=\
_{L}^{\alpha }\mathbf{e}_{i},
\end{equation*}%
where the superposition $\overset{\alpha }{\mathbf{J}}\mathbf{\circ \overset{%
\alpha }{\mathbf{J}}=-I,}$ for $\mathbf{I}$ being the unity matrix. $\ $This
structure is determined by and adapted to N--connection $\ \ _{L}\overset{%
\alpha }{\mathbf{N}}$ induced, in its turn, by \ a regular \ fractional $%
\overset{\alpha }{L}.$

A fractional Lagrangian $\overset{\alpha }{L}$ induces a canonical 1--form
\begin{equation*}
\ \ _{L}^{\alpha }\omega =\frac{1}{2}\left( \ _{\ _{1}y^{i}}\overset{\alpha }%
{\underline{\partial }}_{i}\overset{\alpha }{L}\right) \ ^{\alpha }e^{i}.
\end{equation*}
Following formula $\ _{L}^{\alpha }\mathbf{\theta (X,Y)}\doteqdot \ _{L}
\overset{\alpha }{\mathbf{g}}\left( \overset{\alpha }{\mathbf{J}}\mathbf{X,Y}%
\right) ,$ for any vectors $\mathbf{X}$ and $\mathbf{Y}$ on $\overset{\alpha
}{\underline{T}}M,$ any metric $\ _{L}\overset{\alpha }{\mathbf{g}}$ (\ref%
{sasm}) determines a canonical 2--form
\begin{equation}
\ \ \ _{L}^{\alpha }\mathbf{\theta }=\ _{L\ }\overset{\alpha }{g}_{ij}(x,\
^{\alpha }y)\ \ _{L}^{\alpha }\mathbf{e}^{i}\wedge \ ^{\alpha }e^{j}.
\label{asstr}
\end{equation}

The Main Result in \cite{bv1} (see similar ''integer'' details in \cite%
{ma1,vrflg,vfedq1}) \ states that the fractional canonical metrical
d--connec\-ti\-on $\ _{c}^{\alpha }\mathbf{D}$\textbf{\ }with N--adapted
coefficients (\ref{cdc}), defines a (unique) canonical fractional almost K%
\"{a}hler d--connection $\ _{c}^{\theta }\overset{\alpha }{\mathbf{D}}=$ $\
_{c}^{\alpha }\mathbf{D}$ satisfying the conditions $\ ^{\theta }\overset{%
\alpha }{\mathbf{D}}_{\mathbf{X}}\ \ _{L}\overset{\alpha }{\mathbf{g}}%
\mathbf{=0}$ and $\mathbf{\ }\ \ ^{\theta }\overset{\alpha }{\mathbf{D}}_{%
\mathbf{X}}\overset{\alpha }{\mathbf{J}}=0,$ for any vector $\mathbf{X}%
=X^{i}\ \ _{L}^{\alpha }\mathbf{e}_{i}+X^{a}\ ^{\alpha }e_{a}.$

The Nijenhuis tensor $\overset{\alpha }{\mathbf{\Omega }}$ for $\overset{%
\alpha }{\mathbf{J}}$ is defined in the form%
\begin{equation*}
\overset{\alpha }{\mathbf{\Omega }}(\mathbf{X},\mathbf{Y})\doteqdot \left[
\overset{\alpha }{\mathbf{J}}\mathbf{X,\overset{\alpha }{\mathbf{J}}Y}\right]
-\overset{\alpha }{\mathbf{J}}\left[ \overset{\alpha }{\mathbf{J}}\mathbf{X,Y%
}\right] -\overset{\alpha }{\mathbf{J}}\left[ \mathbf{X,Y}\right] -\left[
\mathbf{X,Y}\right] .
\end{equation*}%
A component calculus with respect to N--adapted bases (\ref{dder}) and (\ref%
{ddif}), for $\ \overset{\alpha }{\mathbf{\Omega }}(\mathbf{e}_{\alpha },%
\mathbf{e}_{\beta })=\overset{\alpha }{\mathbf{\Omega }}_{\alpha \beta
}^{\gamma }\mathbf{e}_{\gamma }$ results in $\overset{\alpha }{\mathbf{%
\Omega }}_{\alpha \beta }^{\gamma }=4\overset{\alpha }{\mathbf{T}}_{\alpha
\beta }^{\gamma },$ where $\overset{\alpha }{\mathbf{T}}_{\alpha \beta
}^{\gamma }$ is the torsion of an affine fractional connection $\overset{%
\alpha }{\mathbf{\Gamma }}_{\alpha \beta }^{\gamma }.$\footnote{%
This formula is a nonholonomic analog, for our conventions, with inverse
sign, of the formula (2.9) from \cite{karabeg1}.} For $\ _{c}^{\alpha }%
\mathbf{D=}\{\ ^{\alpha }\widehat{\mathbf{\Gamma }}_{\ \alpha \beta
}^{\gamma }\}\mathbf{,}$ the components of torsion $\ \ _{L}^{\alpha }%
\widehat{\mathbf{T}}_{\ \alpha \beta }^{\gamma }$ are $\ _{L}^{\alpha }%
\widehat{T}_{jk}^{i}=0,\ _{L}^{\alpha }\widehat{T}_{bc}^{a}=0,\ _{L}^{\alpha
}\widehat{T}_{jk}^{i}=\ _{L}^{\alpha }\widehat{C}_{\ jc}^{i},\ _{L}^{\alpha }%
\widehat{T}_{ij}^{a}=\ _{L}^{\alpha }\Omega _{ij}^{a}, \ _{L}^{\alpha }%
\widehat{T}_{ib}^{a}=\ \ ^{\alpha }e_{b}\ _{L}^{\alpha }N_{i}^{a}-\
_{L}^{\alpha }\widehat{L}_{\ bi}^{a}.$

So, we constructed a canonical (i.e. uniquely determined by $\overset{\alpha
}{L})$ almost K\"{a}hler distinguished connection (d--connection) $\
^{\theta }\overset{\alpha }{\mathbf{D}}$ \ being compatible both with the
almost K\"{a}hler, $\left( \ _{L}^{\alpha }\mathbf{\theta ,}\overset{\alpha }%
{\mathbf{J}}\right) ,$ and N--connection, $\ \ _{L}^{\alpha }\mathbf{N,}$
structures. We can work equivalently with the data $\overset{\alpha }{%
\underline{L^{n}}}=(\overset{\alpha }{\underline{M}},\overset{\alpha }{L}%
)=\left( \ \ _{L}\overset{\alpha }{\mathbf{N}},\ \ _{L}\overset{\alpha }{%
\mathbf{g}},\ _{c}^{\alpha }\mathbf{D}\right) $ and/or $\overset{\alpha }{%
\underline{K^{2n}}}=\left( \overset{\alpha }{\mathbf{J}},\ _{L}^{\alpha }%
\mathbf{\theta ,}\ _{c}^{\alpha }\mathbf{D}\right) .$ The last (nonholonomic
almost symplectic) ones are most convenient for deformation quantization.

\section{Fractional Deformations and Quantization}

\label{s3}In this section we provide a nonholonomic fractional modification
of Fedosov's construction which will be applied for deformation quantization
of fractional Lagrange mechanics, see next section.

\subsection{Star products for fractional symplectic models}

For integer dimensions, any $\ \ \ _{L}^{\alpha }\mathbf{\theta }$ (\ref%
{asstr}) \ induces a structure of Poisson brackets $\{\cdot ,\cdot \}$ via
the Hamilton--Jacobi equations associated to a regular Lagrangian $L,$ see
details in Corollary 2.1 from Ref. \cite{anav}. Working with local
fractional Caputo (co) bases (\ref{frlcb}) and (\ref{frlccb}), the Poisson
structure and derived geometric constructions with data $\overset{\alpha }{%
\underline{K^{2n}}}$ are very similar to those for an abstract,
non--singular, Poisson manifold $(V,\{\cdot ,\cdot \}).$ We shall use the
symbol $V$ for a general space (it can be holonomic, or nonholonomic,
fractional and/or integer etc) in order outline some important concepts
which, for our purposes, will be latter developed for more rich geometric
structures on $V=$ $\overset{\alpha }{\underline{K^{2n}}}.$

Let us denote by $C^{\infty }(V)[[v]]$ the spaces of formal series in
variable $v$ with coefficients from $C^{\infty }(V)$ on a Poisson manifold $%
(V,\{\cdot ,\cdot \}).$ A deformation quantization is an associative algebra
structure on $C^{\infty }(V)[[v]]$ with a $v$--linear and $v$--adically
continuous star product
\begin{equation}
\ ^{1}f\ast \ ^{2}f=\sum\limits_{r=0}^{\infty }\ _{r}C(\ ^{1}f,\ ^{2}f)\
v^{r},  \label{stpr}
\end{equation}%
where $\ _{r}C,r\geq 0,$ are bilinear operators on $C^{\infty }(V)$ with $\
_{0}C(\ ^{1}f,\ ^{2}f)=\ ^{1}f\ ^{2}f$ and $\ _{1}C(\ ^{1}f,\ ^{2}f)-\
_{1}C(\ ^{2}f,\ ^{1}f)=i\{\ ^{1}f,\ ^{2}f\},$ with $i$ being the complex
unity.

If all operators $\ _{r}C,r\geq 0$ are bidifferential, a corresponding star
product $\ast $ is called differential. We can define different star
products on a $(V,\{\cdot ,\cdot \}).$ Two differential star products $\ast $
and $\ast ^{\prime }$ are equivalent if there is an isomorphism of algebras $%
A:\left( C^{\infty }(V)[[v]],\ast \right) \rightarrow \left( C^{\infty
}(V)[[v]],\ast ^{\prime }\right) ,$ where $A=\sum\limits_{r\geq 1}^{\infty
}\ _{r}A\ v^{r},$ for $_{0}A$ being the identity operator and $\ _{r}A$
being differential operators on $C^{\infty }(V).$

\subsection{Fedosov operators for fractional Lagrange spaces}

On $\overset{\alpha }{\underline{K^{2n}}},$ we introduce the tensor $\ \ \
_{L}^{\alpha }\mathbf{\Lambda }^{\beta \gamma }\doteqdot \ \ \ _{L}^{\alpha
}\theta ^{\beta \gamma }-i\ \ \ _{L}^{\alpha }\mathbf{g}^{\beta \gamma },$
where $i$ is the \ complex unity. The local coordinates on $\overset{\alpha }%
{\underline{T}}M$ are parametrized in the form $u=\{u^{\alpha }\}$ and the
local coordinates on $\underline{T}_{u}\overset{\alpha }{\underline{T}}M$ \
are labelled $(u,z)=(u^{\gamma },z^{\beta }),$ where $z^{\beta }$ are the
second order fiber coordinates (we should state additionally a left label $%
\alpha $ if the fractional character of some coordinates has to be
emphasized, for instance to write $\ ^{\alpha }z^{\beta }$ instead of $%
z^{\beta }).$ In deformation quantization, there are used formal series
\begin{equation}
a(v,z)=\sum\limits_{r\geq 0,|\overbrace{\beta }|\geq 0}\ a_{r,\overbrace{%
\beta }}(u)z^{\overbrace{\beta }}\ v^{r},  \label{formser}
\end{equation}%
where $\overbrace{\beta }$ is a multi--index, defining the formal Wick
algebra $\overset{\alpha }{\mathbf{W}}_{u}.$ The formal Wick product $%
\overset{\alpha }{\circ }$ of two elements $a$ and $b$ defined by some
formal series (\ref{formser}) is
\begin{equation}
a\overset{\alpha }{\circ }b\ (z)\doteqdot \exp \left( i\frac{v}{2}\ \ \
_{L}^{\alpha }\mathbf{\Lambda }^{\alpha \beta }\frac{\partial ^{2}}{\partial
z^{\alpha }\partial z_{[1]}^{\alpha }}\right) a(z)b(z_{[1]})\mid
_{z=z_{[1]}}.  \label{fpr}
\end{equation}%
Such a product is determined by a regular fractional Lagrangian $\overset{%
\alpha }{L}$ and corresponding $\overset{\alpha }{\underline{K^{2n}}}.$

Following the constructions from Refs. \cite{vfedq1,vfedq2} for such
''d--algebras'', we construct a nonholonomic bundle $\overset{\alpha }{%
\mathbf{W}}=\cup _{u}$ $\overset{\alpha }{\mathbf{W}}_{u}$ of formal Wick
algebras defined as a union of $\mathbf{W}_{u}.$ The fibre product (\ref{fpr}%
) is trivially extended to the space of $\overset{\alpha }{\mathbf{W}}$%
--valued N--adapted differential forms $\ _{L}^{\alpha }\mathcal{W}\otimes
\overset{\alpha }{\Lambda }$ by means of the usual exterior product of the
scalar forms $\ \ _{L}^{\alpha }\mathbf{\Lambda },$ where $\ _{L}^{\alpha }%
\mathcal{W}$ denotes the sheaf of smooth sections of $\ \overset{\alpha }{%
\mathbf{W}}.$ There is a standard grading on $\ _{L}^{\alpha }\mathbf{%
\Lambda ,}$ denoted $\deg _{a}.$ It is possible to introduce grading $\deg
_{v},\deg _{s},\deg _{a}$ on $\ \ _{L}^{\alpha }\mathcal{W}\otimes \overset{%
\alpha }{\Lambda }$ defined on homogeneous elements $v,z^{\beta },\ \
_{L}^{\alpha }\mathbf{e}^{\beta }$ as follows: $\deg _{v}(v)=1,$ $\deg
_{s}(z^{\alpha })=1,$ $\deg _{a}(\ \ _{L}^{\alpha }\mathbf{e}^{\alpha })=1,$
and all other gradings of the elements $v,z^{\alpha },\ \ _{L}^{\alpha }%
\mathbf{e}^{\alpha }$ are set to zero (we adapt to nonholonomic fractional
configuration the conventions from \cite{karabeg1,vfedq1,vfedq2}). We extend
the canonical d--connection $\ _{c}^{\alpha }\mathbf{D=}\{\ ^{\alpha }%
\widehat{\mathbf{\Gamma }}_{\ \alpha \beta }^{\gamma }\}$ (\ref{cdc})$\ $ to
an operator on $\ _{L}^{\alpha }\mathcal{W}\otimes \overset{\alpha }{\Lambda
}$ following the formula
\begin{equation*}
\ ^{\alpha }\mathbf{\check{D}}\left( a\otimes \lambda \right) \doteqdot
\left( \ \ _{L}^{\alpha }\mathbf{e}_{\alpha }(a)-u^{\beta }\ ^{\alpha }%
\widehat{\mathbf{\Gamma }}_{\ \alpha \beta }^{\gamma }\mathbf{\ }^{z}\mathbf{%
e}_{\alpha }(a)\right) \otimes (\ _{L}^{\alpha }\mathbf{e}^{\alpha }\wedge
\lambda )+a\otimes d\lambda ,
\end{equation*}%
where $^{z}\mathbf{e}_{\alpha }$ is a similar to $\ \ _{L}^{\alpha }\mathbf{e%
}_{\alpha }$ on N--anholonomic fibers of $\underline{T}\overset{\alpha }{%
\underline{T}}M,$ depending on $z$--variables (for holonomic second order
fibers, we can take $^{z}\mathbf{e}_{\alpha }=\partial /\partial z^{\alpha
}).$

\begin{definition}
$\ $The Fedosov distinguished operators (d--operators) $\ \ _{L}^{\alpha
}\delta $ and $\ \ _{L}^{\alpha }\delta ^{-1}$ on $\ \ _{L}^{\alpha }%
\mathcal{W}\otimes \overset{\alpha }{\Lambda }$ are
\begin{eqnarray}
\ \ \ _{L}^{\alpha }\delta (a) &=&\ \ _{L}^{\alpha }\mathbf{e}^{\alpha
}\wedge \mathbf{\ }^{z}\mathbf{e}_{\alpha }(a),  \label{fedop} \\
\ \ \ _{L}^{\alpha }\delta ^{-1}(a) &=&\left\{
\begin{array}{c}
\frac{i}{p+q}z^{\alpha }\ \ _{L}^{\alpha }\mathbf{e}_{\alpha }(a),\mbox{ if }%
p+q>0, \\
{\qquad 0},\mbox{ if }p=q=0,%
\end{array}%
\right.  \notag
\end{eqnarray}%
where $a\in \ _{L}^{\alpha }\mathcal{W}\otimes \overset{\alpha }{\Lambda }$
is homogeneous w.r.t. the grading $\deg _{s}$ and $\deg _{a}$ with $\deg
_{s}(a)=p$ and $\deg _{a}(a)=q.$
\end{definition}

The d--operators (\ref{fedop}) satisfy the property that
\begin{equation*}
a=(\ \ _{L}^{\alpha }\delta \ \ _{L}^{\alpha }\delta ^{-1}+\ \ _{L}^{\alpha
}\delta ^{-1}\ \ _{L}^{\alpha }\delta +\sigma )(a),
\end{equation*}%
where $a\longmapsto \sigma (a)$ is the projection on the $(\deg _{s},\deg
_{a})$--bihomogeneous part of $a$ of degree zero, $\deg _{s}(a)=\deg
_{a}(a)=0.$ We can verify that $\ _{L}^{\alpha }\delta $ is also a $\deg
_{a} $--graded derivation of d--algebra $\left(\ _{L}^{\alpha }\mathcal{W}%
\otimes \overset{\alpha }{\Lambda }\mathbf{,}\overset{\alpha }{\circ }%
\right) .$

A fractional Lagrangian $\overset{\alpha }{L}$ induces respective torsion
and curvature
\begin{eqnarray*}
\widehat{\mathcal{T}} &\doteqdot &\frac{z^{\gamma }}{2}\ \ \ \ _{L}^{\alpha
}\theta _{\gamma \tau }\ _{L}^{\alpha }\widehat{\mathbf{T}}_{\alpha \beta
}^{\tau }(u)\ \ \ _{L}^{\alpha }\mathbf{e}^{\alpha }\wedge \ \ \
_{L}^{\alpha }\mathbf{e}^{\beta }, \\
\widehat{\mathcal{R}} &\doteqdot &\frac{z^{\gamma }z^{\varphi }}{4}\ \
_{L}^{\alpha }\theta _{\gamma \tau }\ _{L}^{\alpha }\widehat{\mathbf{R}}_{\
\varphi \alpha \beta }^{\tau }(u)\ \ \ _{L}^{\alpha }\mathbf{e}^{\alpha
}\wedge \ \ \ _{L}^{\alpha }\mathbf{e}^{\beta },
\end{eqnarray*}%
on $\ _{L}^{\alpha }\mathcal{W}\otimes \overset{\alpha }{\Lambda },$ for $\
_{L}^{\alpha }\widehat{\mathbf{T}}_{\ \alpha \beta }^{\gamma }$ and $\
_{L}^{\alpha }\widehat{\mathbf{R}}_{\ \varphi \alpha \beta }^{\tau }$ being
respectively the torsion and curvature of the canonical d--connection $\
_{c}^{\alpha }\mathbf{D}$ (\ref{cdc}).

Using the formulas (\ref{formser}) and (\ref{fpr}) and the identity
\begin{equation}
\ \ _{L}^{\alpha }\theta _{\varphi \tau }\ _{L}^{\alpha }\widehat{\mathbf{R}}%
_{\ \gamma \alpha \beta }^{\tau }=\ \ _{L}^{\alpha }\theta _{\gamma \tau }\
_{L}^{\alpha }\widehat{\mathbf{R}}_{\ \varphi \alpha \beta }^{\tau },
\label{acp}
\end{equation}%
we prove the important formulas:

\begin{proposition}
The fractional Fedosov d--operators satisfy the properties
\begin{equation}
\left[ \ ^{\alpha }\mathbf{\check{D}},\ \ \ _{L}^{\alpha }\delta \right] =%
\frac{i}{v}ad_{Wick}(\widehat{\mathcal{T}})\mbox{ and }\ \ ^{\alpha }\mathbf{%
\check{D}}^{2}=-\frac{i}{v}ad_{Wick}(\widehat{\mathcal{R}}),  \label{comf}
\end{equation}%
where $[\cdot ,\cdot ]$ is the $\deg _{a}$--graded commutator of
endomorphisms of $\ \ \ _{L}^{\alpha }\mathcal{W}\otimes \overset{\alpha }{%
\Lambda }$ and $ad_{Wick}$ is defined via the $\deg _{a}$--graded commutator
in $\left( \ \ _{L}^{\alpha }\mathcal{W}\otimes \overset{\alpha }{\Lambda }%
\mathbf{,}\overset{\alpha }{\circ }\right) .$
\end{proposition}

We have all formal geometric ingradients for performing deformation
quantization of fractional Lagrange mechanics.

\section{Fedosov Quantization of Fractional Mechanics}

\label{s4}The Fedosov's deformation quantization theory is generalized to
fractional Lagrange spaces. The class $c_{0}$ of the deformation
quantization of fractional Lagrange geometry is calculated.

\subsection{Main theorems for fractional Lagrange spaces}

We denote the $Deg$--homogeneous component of degree $k$ of an element $a\in
\ \ _{L}^{\alpha }\mathcal{W}\otimes \overset{\alpha }{\Lambda }$ by $%
a^{(k)}.$

\begin{theorem}
\label{th2}For any regular fractional Lagrangian $\overset{\alpha }{L}$ and
correspoinding canonical almost K\"{a}hler--Lagrange model $\overset{\alpha }%
{\underline{K^{2n}}}=\left( \overset{\alpha }{\mathbf{J}},\ _{L}^{\alpha }%
\mathbf{\theta ,}\ _{c}^{\alpha }\mathbf{D}\right) ,$ there is a flat
canonical fractional Fedosov d--connection
\begin{equation*}
\ \ \ _{L}^{\alpha }\widehat{\mathcal{D}}\doteqdot -\ \ \ _{L}^{\alpha
}\delta +\ ^{\alpha }\mathbf{\check{D}}-\frac{i}{v}ad_{Wick}(r)
\end{equation*}%
satisfying the condition $\ \ \ _{L}^{\alpha }\widehat{\mathcal{D}}^{2}=0,$
where the unique element $r\in $ $\ \ _{L}^{\alpha }\mathcal{W}\otimes
\overset{\alpha }{\Lambda }\mathbf{,}$ $\deg _{a}(r)=1,$ $\ \ ^{\alpha }%
\mathbf{\check{D}}\ \ _{L}^{\alpha }\delta ^{-1}r=0,$ solves the equation
\begin{equation*}
\ \ \ _{L}^{\alpha }\delta r=\widehat{\mathcal{T}}+\widehat{\mathcal{R}}+\
^{\alpha }\mathbf{\check{D}}r-\frac{i}{v}r\overset{\alpha }{\circ }r
\end{equation*}%
and this element can be computed recursively with respect to the total
degree $Deg$ as follows:%
\begin{eqnarray*}
r^{(0)} &=&r^{(1)}=0,\ r^{(2)}=\ \ \ _{L}^{\alpha }\delta ^{-1}\widehat{%
\mathcal{T}}, \\
r^{(3)} &=&\ \ \ _{L}^{\alpha }\delta ^{-1}\left( \widehat{\mathcal{R}}+\
^{\alpha }\mathbf{\check{D}}r^{(2)}-\frac{i}{v}r^{(2)}\overset{\alpha }{%
\circ }r^{(2)}\right) , \\
r^{(k+3)} &=&\ \ \ _{L}^{\alpha }\delta ^{-1}\left( \ ^{\alpha }\mathbf{%
\check{D}}r^{(k+2)}-\frac{i}{v}\sum\limits_{l=0}^{k}r^{(l+2)}\overset{\alpha
}{\circ }r^{(l+2)}\right) ,k\geq 1.
\end{eqnarray*}
\end{theorem}

\begin{proof}
The proof is similar to the standard Fedosov constructions if we work with
the Caputo fractional derivative in N--adapted form, by induction using the
identities%
\begin{equation*}
\ \ \ _{L}^{\alpha }\delta \widehat{\mathcal{T}}=0\mbox{ and }\ \ \
_{L}^{\alpha }\delta \widehat{\mathcal{R}}=\ \ ^{\alpha }\mathbf{\check{D}}%
\widehat{\mathcal{T}}.
\end{equation*}%
For integer dimensions and holonomic configurations we get the results from
Ref. \cite{karabeg1} proved for arbitrary affine connections with torsion
and almost K\"{a}hler structures on $M.$ $\Box $
\end{proof}

\vskip3pt

The next theorem gives a rule how to define and compute the star product
(which is the main purpose of deformation quantization) induced by a regular
fractional Lagrangian.

\begin{theorem}
\label{th3}A star--product $\overset{\alpha }{\ast }$ on the canonical
almost K\"{a}hler model of fractional Lagrange space $\overset{\alpha }{%
\underline{K^{2n}}}=\left( \overset{\alpha }{\mathbf{J}},\ _{L}^{\alpha }%
\mathbf{\theta ,}\ _{c}^{\alpha }\mathbf{D}\right) $ is defined on $%
C^{\infty }(\overset{\alpha }{\underline{L^{n}}})[[v]]$ by formula
\begin{equation*}
\ ^{1}f\overset{\alpha }{\ast }\ ^{2}f\doteqdot \sigma (\tau (\ ^{1}f))%
\overset{\alpha }{\circ }\sigma (\tau (\ ^{2}f)),
\end{equation*}%
where the projection $\sigma :\ \ \ _{L}^{\alpha }\mathcal{W}_{\widehat{%
\mathcal{D}}}\rightarrow C^{\infty }(\overset{\alpha }{\underline{L^{n}}}%
)[[v]]$ onto the part of $\deg _{s}$--degree zero is a bijection and the
inverse map $\tau :C^{\infty }(\overset{\alpha }{\underline{L^{n}}}%
)[[v]]\rightarrow \ _{L}^{\alpha }\mathcal{W}_{\widehat{\mathcal{D}}}$ \
is calculated recursively w.r.t. the total degree $Deg,$%
\begin{eqnarray*}
\tau (f)^{(0)} &=&f\mbox{\ and, for }k\geq 0, \\
\tau (f)^{(k+1)} &=&\ \ \ _{L}^{\alpha }\delta ^{-1}\left( \ \ \ ^{\alpha }%
\mathbf{\check{D}}\tau (f)^{(k)}-\frac{i}{v}\sum%
\limits_{l=0}^{k}ad_{Wick}(r^{(l+2)})(\tau (f)^{(k-l)})\right) .
\end{eqnarray*}
\end{theorem}

\begin{proof}
We may check by explicit computations similarly to those in \cite{karabeg1},
in our case, with fractional Caputo derivatives that such constructions give
a well defined star product. $\Box $
\end{proof}

\subsection{Cohomology classes of quantized fractional Lagrangians}

The characteristic class of a star product  is $(1/iv)[\theta
]-(1/2i)\varepsilon ,$ where $\varepsilon $ is the canonical class for an
underlying K\"{a}hler manifold, for nonholonomic Lagrange--Einstein--Finsler
spaces  we analysed this construction in Refs. \cite%
{vfedq1,vfedq2,vfedq3,vfed4}. This canonical class can be defined for any
almost complex manifold. We show how such a calculus of the crucial part of
the characteristic class $cl$ of the fractional star product $\overset{%
\alpha }{\ast }$ from  Theorem \ref{th3} can be performed. In explicit form,
we shall compute the coefficient $c_{0}$ at the zeroth degree of $v.$

A straightforward computation of $\ _{2}C$ from (\ref{stpr}), using
statements of Theorem \ref{th2}, results in a proof of

\begin{lemma}
\label{lem1}There is a unique fractional  2--form $\ _{L}^{\alpha
}\varkappa $ which can be expressed
\begin{equation*}
\ \ _{L}^{\alpha }\varkappa =-\frac{i}{8}\ _{L}^{\alpha }\mathbf{J}_{\tau
}^{\ \gamma ^{\prime }}\ _{L}^{\alpha }\widehat{\mathbf{R}}_{\ \gamma
^{\prime }\gamma \beta }^{\tau }\ _{L}^{\alpha }\mathbf{e}^{\gamma }\wedge \
_{L}^{\alpha }\mathbf{e}^{\beta }-i\ \ _{L}^{\alpha }\lambda ,
\end{equation*}%
where the exact  1--form $\ _{L}^{\alpha }\lambda =d \
_{L}^{\alpha }\mu ,$ for $\ _{L}^{\alpha }\mu =\frac{1}{6}\ _{L}^{\alpha }%
\mathbf{J}_{\tau }^{\ \alpha ^{\prime }}\ _{L}^{\alpha }\widehat{\mathbf{T}}%
_{\ \alpha ^{\prime }\beta }^{\tau }\ _{L}^{\alpha }\mathbf{e}^{\beta },$
with nontrivial components of curvature and torsion defined by the canonical d--connection.
\end{lemma}

This allows us to compute the Chern--Weyl form
\begin{eqnarray*}
 \ _{L}^{\alpha }\gamma  &=&-iTr\left[ \left( h\Pi ,v\Pi \right) \
_{L}^{\alpha }\widehat{\mathbf{R}}\left( h\Pi ,v\Pi \right) ^{T}\right] =-iTr%
\left[ \left( h\Pi ,v\Pi \right) \ _{L}^{\alpha }\widehat{\mathbf{R}}\right]
\\
&=&-\frac{1}{4}\ _{L}^{\alpha }\mathbf{J}\ _{\tau }^{\ \alpha ^{\prime }}\
_{L}^{\alpha }\widehat{\mathbf{R}}_{\ \alpha ^{\prime }\alpha \beta }^{\tau
}\ _{L}^{\alpha }\mathbf{e}^{\alpha }\wedge \ _{L}^{\alpha }\mathbf{e}%
^{\beta }
\end{eqnarray*}%
to be closed. By definition, the canonical class is $\ _{L}^{\alpha
}\varepsilon \doteqdot \lbrack \ \ _{L}^{\alpha }\gamma ].$\footnote{%
For simplicity, we recall the definition of the canonical class $\varepsilon
$ of an almost complex manifold $(M,\mathbb{J})$ of integer dimension and
redefine it for $\ ^{N}TTM=hTM\oplus vTM.$  The
distinguished complexification of such second order tangent bundles is
introduced in the form
 $T_{\mathbb{C}}\left(\ ^{N}TTM\right) =T_{\mathbb{C}}\left( hTM\right)
\oplus T_{\mathbb{C}}\left( vTM\right).$ %
For such nonholonomic bundles, the class $\ ^{N}\varepsilon $ is the first
Chern class of the distributions $T_{\mathbb{C}}^{\prime }\left( \
^{N}TTM\right) =T_{\mathbb{C}}^{\prime }\left( hTM\right) \oplus T_{\mathbb{C%
}}^{\prime }\left( vTM\right) $ of couples of vectors of type $(1,0)$ both
for the h-- and v--parts. We can calculate both for integer and fractional
dimensions the canonical class $^{L}\varepsilon $ (we put the label $L$ for
the constructions canonically defined by a regular Lagrangian $L)$ for the
almost K\"{a}hler model of a Lagrange space $L^{n}.$ We take the canonical
d--connection $\ ^{L}\widehat{\mathbf{D}}$ that it was used for constructing
$\ast $ and considers h- and v--projections
 $h\Pi =\frac{1}{2}(Id_{h}-iJ_{h})$ and
 $v\Pi =\frac{1}{2}(Id_{v}-iJ_{v}),$
where $Id_{h}$ and $Id_{v}$ are respective identity operators and $J_{h}$
and $J_{v}$ are  projection
operators onto corresponding $(1,0)$--subspaces. The matrix $\left( h\Pi
,v\Pi \right) \widehat{\mathbf{R}}\left( h\Pi ,v\Pi \right) ^{T},$ where $%
(...)^{T}$ denotes the transposition, is the curvature matrix of the
restriction of the connection $\ ^{L}\widehat{\mathbf{D}}$ to $T_{\mathbb{C}%
}^{\prime }\left( \ ^{N}TTM\right) .$ For fractional dimensions, such
formulas ''obtain'' corresponding let labels with $\alpha .$} These formulas
and  Lemma \ref{lem1} give the proof of
\begin{theorem}
The zero--degree cohomology coefficient
for the almost K\"{a}hler model of fractional Lagrange space  $\overset{\alpha }{\underline{L^{n}}}$ is computed
 $c_{0}(\overset{\alpha }{\ast })=-(1/2i)\ \ _{L}^{\alpha }\varepsilon ,$
 where the value $\ _{L}^{\alpha }\varepsilon $ is canonically defined by a
regular fractional  Lagrangian $\overset{\alpha }{L}(u).$
\end{theorem}

Finally we note that the formula from this Theorem can be directly applied
for the Cartan connection in Finsler geometry with $\overset{\alpha }{L}%
=\left( \overset{\alpha }{F}\right) ^{2},$ where $\overset{\alpha }{F}$ is the fundamental generating Finsler function in fractional Finsler geometry, and in certain fractional generalizations of Einstein and Ricci flow theories \cite{vrfrf,vrfrg}.

\section{Conclusions and Discussion}

\label{s5} In this paper we provided a generalization of Fedosov's method for quantizing the fractional Lagrange mechanics with Caputo fractional derivatives. We used a fundamental result that nonholonomic geometries (for certain classes of integro--differential distributions modeling fractional spaces \cite{vrfrf,vrfrg}) can be modeled as some almost K\"{a}hler configurations which can be quantized following  Karabegov and Schlichenmaier  ideas \cite{karabeg1}.

We argue that the approach to fractional calculus based on Caputo fractional derivative is a self--consistent comprehensive one allowing geometrization of fundamental field and evolution equations and their quantization at least in the meaning of deformation quantization theory.

In various directions of modern mathematics, physics, mathematical economics etc, there are also considered, and preferred, different fractional derivatives, for instance, the Riemann--Liouville (RL) derivative. It is a problem, at least technically,  to elaborate a well defined differential geometry with RL type fractional derivatives not resulting in zero acting on constants (see detailed discussions in \cite{vrfrf,vrfrg}). So, for such fractional calculus approaches we can not geometrize mechanical and field/evolution interactions in a standard form. In general, it is not clear how to define a RL--differential geometry which would mimic certain integer dimention type geometries.  As a result, we can not perform a RL--quantization following usual geometric/deformation methods.

Our proposal, is that for fractional models, for instance, with RL fractional derivative, we can geometrize the constructions, and elaborate quantum models taking the Caputo derivatives for certain background constructions and then to deform nonholonomically the geometric objects in order to re--adapt them and generate a necessary RL, or another ones, fractional theory.

\vskip5pt \textbf{Acknowledgement: } S. V. is grateful to \c{C}ankaya
University for support of his research on fractional calculus, geometry and
applications.

\end{document}